\documentclass[a4paper, 11pt]{article}
\usepackage{amssymb}
\usepackage{amsmath}
\usepackage{amstext}
\usepackage{amsfonts}
\usepackage{amsthm}
\usepackage{mathrsfs}
\usepackage{stmaryrd}
\usepackage{graphicx}
\usepackage{color}
\usepackage{tikz}
\usepackage{pgfplots}
\usepackage[toc,page]{appendix}
\usepackage[norefpage,intoc]{nomencl} 
\usepackage[colorlinks=true, citecolor=red, linkcolor=blue, urlcolor=blue]{hyperref}
\usepackage{marvosym}
\makenomenclature

\newcommand{\R}{\mathbb{R}}
\newcommand{\ZZ}{\mathbb{Z}}

\newtheorem{thm}{Theorem}
\newtheorem{cor}[thm]{Corollary}
\newtheorem{lemma}[thm]{Lemma}
\newtheorem{prop}[thm]{Proposition}

\theoremstyle{definition}
\newtheorem{rem}[thm]{Remark}

\title{Approximation of eigenvalues of Schr\"odinger operators}

\author{Johannes F. Brasche, Robert Fulsche}

\pgfplotsset{compat=1.13}
\begin{document}
\maketitle
\begin{abstract}
It is known that convergence of l.s.b. closed symmetric sesquilinear forms implies norm resolvent convergence of the associated self-adjoint operators and this in turn
convergence of discrete spectra. In this paper in both cases sharp estimates for the rate of convergence are derived.
An algorithm for the numerical computation of eigenvalues of 
generalized Schr{\"o}dinger operators in $L^2(\R)$ is presented and illustrated by explicit examples; 
the mentioned general results on the rate of convergence are applied in order to
obtain error estimates for these computations. An extension of the results to Schr{\"o}dinger operators on metric graphs is sketched.  
\end{abstract}
\begin{flushleft}

\end{flushleft}
\section{Introduction}
%The computation of negative eigenvalues of Schr\"{o}dinger operators is of everlasting interest in mathematical physics. Since the eigenvalue equation cannot be %solved analytically for most potentials, numerical methods come into play here. Well-known approaches to numerical approximation of negative eigenvalues include %those based on variational formulations for the eigenvalues, e.g. the Rayleigh-Ritz method. We suggest another approach for Schr\"odinger operators on %$L^2(\mathbb{R})$. 

\noindent
Often things become easier by passing to the limit. A striking example is provided by the one-dimensional Kronig-Penney model. Here the Hamiltonian
is the self-adjoint operator $H_b= -\Delta + b\sum_{n\in \ZZ} a_n \delta_n$ in $L^2(\R)$. If all coefficients $a_n$ are positive, then, by Kato's
monotone convergence theorem, the operators $H_b$ converge in the strong resolvent sense to the Laplacian with Dirichlet boundary conditions at
every point of $\ZZ$, as the coupling parameter $b$ tends to infinity. If the family $(a_n)_{n\in \ZZ}$ is lower bounded by a positive
constant, then the operators $H_b$ converge in the norm resolvent sense with convergence rate $O(1/b)$ \cite[Example 3.8]{hk}. Moreover
if the family $(1/a_n)_{n\in \ZZ}$ is summable, then the resolvents even converge w.r.t. the trace norm with convergence rate $O(1/b)$
\cite[Theorem 3.1]{mfat}.
While it is difficult to investigate the operator $H_b$, it is trivial
to determine the spectral properties of the limit operator and results on the rate of convergence may be used
for a discussion of the Hamiltonian $H_b$. We refer to \cite{robin,hk,lms,jmaa,jfa,mfat,bd,gru} for a detailed analysis of large coupling convergence of regular and generalized Schr{\"o}dinger operators as well as other kinds of operators. 

\noindent
Another important class of examples is provided by Schr{\"o}dinger operators with singular potentials. If the potential $V$ vanishes outside
a very small vicinity of the closed null set $N$, then one expects that a replacement of the Hamiltonian $-\Delta +V$ by a suitably chosen
self-adjoint operator $H$ that coincides with $-\Delta$ on the space $C_0^{\infty}(\R^d\setminus N)$ only leads to a negligible error. Due to the fact that $H$ is equal to the free Hamiltonian outside the null set $N$ it is often easier to investigate
the operator $H$ than the regular Schr{\"o}dinger operator $-\Delta +V$. The idea that the mentioned replacement only leads to small errors
is supported by a large variety of convergence results. The most famous ones are results on point interactions; if the set $N$ is finite and $d=2,3$, then for every self-adjoint operator
$H$ in $L^2(\R^d)$ that coincides with $-\Delta$ on the space $C_0^{\infty}(\R^d\setminus N)$
one can give 
potentials $V_n$ such that the operators $-\Delta+V_n$ converge in the norm resolvent sense to $H$ and the supports of the potentials
$V_n$ shrink to the set $N$ \cite{Albeverio}. 

\noindent
Originally point interaction Hamiltonians have been introduced as an idealization of certain Schr{\"o}dinger operators with short range potential. 
In \cite{Brasche} one had the idea to use such operators for other purposes, too. One has shown that a large class of operators $H$ can be approximated
by point interaction Hamiltonians $H_n$. Since it is easy to compute the eigenvalues of the operators $H_n$, this fact may be used in order to 
compute the eigenvalues of $H$ approximately. Later a modification of the methods of \cite{Brasche} has led to an efficient algorithm for the 
numerical computation of eigenvalues of regular and generalized Schr{\"o}dinger operators in $L^2(\R^d)$, $d=2,3$ \cite{bo}. In the subsequent article 
\cite{o} these results have been extended to Schr{\"o}dinger operators with magnetic potentials. 

\noindent
In the present paper we
shall concentrate on the one-dimensional case where it is possible to obtain faster convergence with simpler algorithms.   
It was shown in \cite{Brasche} that weak convergence of measures implies norm resolvent convergence of the Schr\"odinger operators having these measures as potentials (cf. Theorem \ref{formconvergence} below). Based on this convergence result, we give a general construction for approximating finite signed measures $\mu$ by pure point measures $\mu_n$ such that the operators $-\Delta + \mu_n$ converge to $-\Delta + \mu$. Since norm resolvent convergence implies convergence of the isolated eigenvalues and the eigenvalues of $-\Delta + \mu_n$ can be computed efficiently, we end up with a general method for approximating negative eigenvalues of Schr\"odinger operators with integrable potential on $\mathbb{R}$. We provide error estimates for each step of the construction and proof.\\
This paper is organized as follows: In section 2 we recall some notation and preliminary results. Section 3 and 4 provide very general results:
As it is well known convergence of l.s.b. closed symmetric sesquilinear forms implies convergence of the associated self-adjoint operators in the norm resolvent sense.
In section 3 we derive an asymptotically sharp estimate on the rate of convergence. 
Norm resolvent convergence of self-adjoint operators implies convergence of the points in the discrete spectra. In
 section 4 we quantify this result. Section 5 gives results on the convergence of the sesquilinear forms for our concrete problem. In section 6 we discuss the actual construction of the approximating measures, including error estimates in terms of the Fourier transform. Section 7 shows an efficient way to compute negative eigenvalues of Schr\"odinger operators with pure point potentials. In section 8 we added a short discussion about existence results of negative eigenvalues for Schr\"odinger operators. Section 9 demonstrates our algorithm by two examples. In one of the two examples we deal with a singularly continuous measure potential, which cannot be treated by the classical approximation methods. Finally, in section 10 we discuss how our method may be extended to Schr\"odinger operators on other 1-dimensional domains and explicitly show this for operators on $[0, \infty)$.
\section{Preliminaries}
We will always denote by $\mathcal{H}$ a Hilbert space with inner product $\langle \cdot, \cdot \rangle$, which is linear in the first and antilinear in the second component. For a sesquilinear form $a$ on $\mathcal{H}$ with domain $D(a)$ and for $c \in \mathbb{R}$ we denote by $a_c$ the form
\begin{align*}
D(a_c) &= D(a)\\
a_c(f,g) &= a(f,g) + c\langle f, g\rangle.
\end{align*}
If $c$ is a lower bound of $a$, then $a_{1-c}$ is an inner product on $D(a)$. Further, for $f \in D(a)$ we sometimes denote
\begin{align*}
a[f] := a(f,f).
\end{align*}
The well known Sobolev inequality for the Sobolev space $H^1(\mathbb{R})$ states that for each $\varepsilon > 0$ and each $f \in H^1(\mathbb{R})$
\begin{align}\label{sobolevineq}
\| f\|_\infty^2 \leq \varepsilon \| f'\|_{L^2}^2 + \frac{4}{\varepsilon} \| f\|_{L^2}^2.
\end{align}
Let $\mu$ be a finite Radon measure on $\mathbb{R}$, i.e. a finite signed measure on the Borel-$\sigma$-algebra of $\mathbb{R}$, $\mathcal{B}(\mathbb{R})$. Then we define the sesquilinear form $a_\mu$ by
\begin{align*}
D(a_\mu) &= H^1(\mathbb{R})\\
a_\mu(f,g) &= \int_\mathbb{R}f'(x) \overline{g'(x)}dx + \int_\mathbb{R} f \overline{g} d\mu,
\end{align*}
which is well known to be symmetric, lower-semibounded and closed. By $-\Delta + \mu$ we denote the self-adjoint operator associated to the form in the sense of Kato's first representation theorem (cf. \cite[Theorem 2.6]{Kato}, \cite[Theorem VIII.15]{Reed1}), i.e. $-\Delta + \mu$ is the self-adjoint operator in $L^2(\mathbb{R})$ with 
\begin{align*}
D(-\Delta +\mu) &\subseteq D(a_\mu)\\
\langle (-\Delta + \mu) f, g\rangle &= a_\mu(f,g) \quad \forall f \in D(-\Delta + \mu), g \in D(a_\mu).
\end{align*}
We say that a sequence $(\mu_n)$ of finite Radon measures on $\mathbb{R}$ converges weakly to the finite Radon measure $\mu$ if
\begin{align*}
\int_\mathbb{R} f(x) d\mu_n(x) \to \int_\mathbb{R} f(x) d\mu(x), \quad n \to \infty
\end{align*}
for each $f \in C_b(\mathbb{R})$, where $C_b(\mathbb{R})$ is the space of bounded continuous functions on $\mathbb{R}$. In particular, the corresponding Fourier transforms, which are defined as
\begin{align*}
\hat{\nu}(t) := \int_\mathbb{R} e^{itx} d\nu(x)
\end{align*}
for a finite Radon measure $\nu$ on $\mathbb{R}$, converge pointwise. For a finite Radon measure $\mu$ we let $\mu = \mu_+ - \mu_-$ be the Hahn-Jordan decomposition and set $|\mu| = \mu_+ + \mu_-$. If $f$ is a measurable function, we define the measure $f\mu$ by
\begin{align*}
f\mu(B) = \int_B fd\mu, \quad B \in \mathcal{B}(\mathbb{R}).
\end{align*}
$\chi_B$ denotes the indicator function of $B \in \mathcal{B}(\mathbb{R})$. The following theorem plays a key role in this paper: 
\begin{thm}[{\cite[Theorem 3]{Brasche}}]\label{formconvergence}
Let $\mu_n$, $n \in \mathbb{N}$, and $\mu$ be finite Radon measures on $\mathbb{R}$ such that $\mu_n \to \mu$ weakly. Then the operators $-\Delta + \mu_n$ converge to $-\Delta + \mu$ in norm resolvent sense.
\end{thm}
One goal of the next sections will be to quantify this theorem.
\section{Form convergence and norm resolvent convergence}
The fact that the convergence of sesquilinear forms with common domain implies norm resolvent convergence of the associated operators is well known, cf. \cite[Theorem VIII.25]{Reed1}. The following theorem gives a quantitative result for this statement.
\begin{thm}
Let $A$ and $B$ be self-adjoint and lower semibounded linear operators in $\mathcal{H}$ such that $A \geq 1$ and $B \geq 1$. 
Let $a$ and $b$ be the closed sesquilinear form associated to $A$ and $B$, respectively. Assume that 
%Further assume that, letting $a$ be the closed sesquilinear form associated to $A$ and $b$ be the closed sesquilinear form associated to $B$, 
$D(a) = D(b) =: D$. If
\[ s := \sup_{g \in D, \  a(g, g) = 1} |a(g, g)  - b(g, g)| < 1, \]
then
\[ \| B^{-1} - A^{-1}\| \leq \frac{s}{\sqrt{1-s}}. \]
\end{thm}
\begin{proof}
For the whole proof let always $f$ be arbitrary in $\mathcal{H}$ such that $\| f\|=1$. Since $B \geq 1$,
\begin{align}\label{Eq1}
\| B^{-1} f\| \leq 1.
\end{align}
Set
\begin{align*}
g = \frac{B^{-1}f}{\sqrt{a(B^{-1}f, B^{-1}f)}}.
\end{align*}
Then $a(g, g) = 1$ and therefore
\begin{align*}
|a(g,g) - b(g,g)| \leq s.
\end{align*}
Further,
\begin{align*}
b(g,g) &= \frac{\langle f, B^{-1}f\rangle}{a(B^{-1}f, B^{-1}f)}\\
&\leq \frac{\| f\| \|B^{-1} f\|}{a(B^{-1} f, B^{-1} f)}\\
&\leq \frac{1}{a(B^{-1} f, B^{-1} f)}
\end{align*}
by inequality (\ref{Eq1}). Since $s \geq 1 - b(g,g)$, we can conclude
\begin{align*}
s \geq 1 - \frac{1}{a(B^{-1} f, B^{-1} f)}
\end{align*}
and therefore
\begin{align}
a(B^{-1} f, B^{-1} f) \leq \frac{1}{1-s}.
\end{align}
Using $A\geq 1$ (and therefore $a\geq 1$ as well) we get
\begin{align*}
\| B^{-1} f - A^{-1} f\| &\leq \sqrt{a(B^{-1} f - A^{-1} f, B^{-1} f - A^{-1} f)}.
\end{align*}
Denoting $d = B^{-1} f - A^{-1} f$, we have
\begin{align*}
\| B^{-1} f - A^{-1} f\| &\leq \sqrt{a(d,d)}\\
&= \sup_{a(h,h)=1} |a(d,h)|\\
&= \sup_{a(h,h)=1} |a(B^{-1}f, h) - \langle f, h\rangle|\\
&= \sup_{a(h,h)=1} |a(B^{-1}f, h) - b(B^{-1}f, h)|\\
&= \sup_{a(h,h)=1} \Big | a \Big ( \frac{B^{-1}f}{a[B^{-1}f]^{1/2}}, h \Big ) - b \Big ( \frac{B^{-1}f}{a[B^{-1}f]^{1/2}}, h \Big ) \Big | a[B^{-1}f]^{1/2}\\
&\leq \sup_{a(c,c) = 1=a(h,h)} |a(c,h) - b(c,h)|\frac{1}{\sqrt{1-s}}.
\end{align*}
The well-known identity
\begin{align*}
\sup_{a(c,c) = 1 = a(h,h)} |a(c,h) - b(c,h)| = \sup_{a(h,h)=1} |a(h,h)-b(h,h)|
\end{align*}
completes the proof.
\end{proof}
\begin{rem}
Observe that the above estimate is good in the sense that it is asymptotically sharp: In the trivial case 
$a=\langle \cdot,\cdot \rangle$ and $b=(1+s)a$ we have 
\[ \| B^{-1}-A^{-1} \| = \frac{s}{1+s}, \]
and the quotient of $\| B^{-1} - A^{-1} \| $ and the upper bound for $\| B^{-1} - A^{-1} \| $ in the above
estimate tends to $1$, as $s$ tends to zero. 
%On some Hilbert space $\mathcal{H}$ 
%let $A = I$ and $A_n = (1 + \frac{1}{n}) I$. 
%Let $a$ and $a_n$ be the associated sesquilinear forms, i.e. $a$ is just the inner product of the Hilbert space and $a_n$ a multiple of it. Then
%\begin{align*}
%\| A^{-1} - A_n^{-1}\| &= \| I - \frac{n}{n+1}I\| = \frac{1}{n+1}
%\end{align*}
%and
%\begin{align*}
%s = \sup_{a(f,f) = 1} | a(f,f) - a_n(f,f)| = \frac{1}{n},
%\end{align*}
%thus the theorem gives
%\begin{align*}
%\frac{1}{n+1} \leq \frac{1}{\sqrt{n(n-1)}}
%\end{align*}
%where the quotient of both sides converges to 1 for $n \to \infty$.
\end{rem}
\section{Convergence of isolated eigenvalues}
A well-known consequence of norm convergence and norm resolvent convergence for self-adjoint operators is that it implies convergence of the discrete spectrum. The aim of the next two results is to quantify this result.
\begin{lemma}
Let $S$ and $T$ be bounded and self-adjoint operators in $\mathcal{H}$ satisfying $\| S - T\| < \delta$ for some $\delta > 0$. If
\begin{align*}
(E - 2\delta, E + 2\delta) \cap \sigma(S) = \{ E\}
\end{align*}
and $E$ is an eigenvalue of $S$ of multiplicity $n \in \mathbb{N}$, then
\begin{align*}
(E - \delta, E + \delta) \cap \sigma(T) \subset \sigma_d(T)
\end{align*}
and the number, counting multiplicities, of eigenvalues of $T$ in $(E-\delta,E+\delta)$ is equal to $n$. 
%consists only of finitely many eigenvalues with summed multiplicity $n$.
\end{lemma}
\begin{proof}
For each $f \in \mathcal{H}$ and $B \in \mathcal{B}(\mathbb{R})$ we use the notation $\mu_{f,A}(B) = \| \chi_B(A) f\|^2$ for each self-adjoint operator $A$ on $\mathcal{H}$. It is well-known that $\mu_{f,A}$ is a measure satisfying $\mu_{f,A}(\mathbb{R}) = \| f\|^2$ and $\| (A-c)f\|^2 = \int |\lambda - c|^2 d\mu_{f,A}(\lambda)$ for $c \in \mathbb{R}$. Further, $\mathcal{H}$ decomposes into $\mathcal{H} = \operatorname{ran} 1_B(A) \oplus \operatorname{ran} 1_{\mathbb{R} \setminus B} (A)$ for each $B \in \mathcal{B}(\mathbb{R})$.\\
Now let $J = (E - \delta, E + \delta)$ and $\widetilde{J} = (E - 2\delta, E + 2\delta)$. The statement we want to prove is equivalent to $\operatorname{dim} \operatorname{ran} 1_J(T) = n$. We will prove this by contradiction. \\
Case 1: Assume $\operatorname{dim} \operatorname{ran} 1_J(T) < n$. Since $\operatorname{dim} \operatorname{ker} (S-E) = n > \operatorname{dim} \operatorname{ran} 1_J(T)$, there is some $f \in \operatorname{ker}(S-E)$ with $\| f\| = 1$ and $f \perp \operatorname{ran} 1_J(T)$. Hence,
\begin{align*}
\mu_{f,T}(J) = \| 1_J(T) f\|^2 = \langle 1_J(T)^2 f, f\rangle = 0
\end{align*}
and thus
\begin{align*}
\| (T-E)f\|^2 = \int |\lambda - E|^2 d\mu_{f,T}(\lambda) \geq \int \delta^2 \mu_{f,T}(\lambda) = \delta^2 \| f\|^2.
\end{align*}
This is a contradiction, since it would imply
\begin{align*}
\| (T-S)f\| = \| (T-E)f - (S-E)f\| = \| (T-E)f\| \geq \delta.
\end{align*}
Case 2: $\operatorname{dim} \operatorname{ran} 1_J(T) > n$. In this case we can choose $f \in \operatorname{ran} 1_J(T)$ with $\| f\| = 1$ and $f \perp \operatorname{ker}(S-E)$. For this $f$ it holds $f \perp \operatorname{ran}_{\mathbb{R} \setminus J}(T)$ and hence
\begin{align*}
\mu_{f,T}(\mathbb{R} \setminus J) = \langle (1_{\mathbb{R} \setminus J}(T))^2 f, f\rangle = 0.
\end{align*}
Thus,
\begin{align*}
\| (T-E)f\|^2 = \int |\lambda - E|^2 d\mu_{f,T}(\lambda) \leq \int \delta^2 d\mu_{f,T}(\lambda) = \delta^2 \| f\|^2.
\end{align*}
It is $\operatorname{ran}1_{\widetilde{J}}(S) = \operatorname{ker}(S-E)$, since $E$ is the only point in $\sigma(S) \cap \widetilde{J}$.
This implies
\begin{align*}
\mu_{f,S} (\widetilde{J}) = \langle (1_{\widetilde{J}}(S))^2 f, f\rangle = 0 
\end{align*}
and hence
\begin{align*}
\| (S-E)f\|^2 = \int |\lambda - E|^2 d\mu_{f,S}(\lambda) \geq \int (2\delta)^2 d\mu_{f,S}(\lambda) = (2\delta)^2 \| f\|^2.
\end{align*}
However, this is impossible, since it would imply
\[ \| (S-T)f\| = \| (S-E) f - (T-E)f\| \geq \|(S-E)f\| - \|(T-E)f\| \geq \delta. \qedhere \]
\end{proof}
Recall the following result:
\begin{thm}[{\cite[Remark 2.1 and Theorem 3.1]{Brasche2}}]\label{essspectrum}
For every finite Radon measure $\mu$ on $\mathbb{R}$ it holds
\[ \sigma_{ess}(-\Delta +\mu) = [0, \infty). \]
\end{thm}
From the previous lemma one can derive the following result, which is applicable to our setting:
\begin{thm}
Let $A, \widetilde{A}$ be self-adjoint operators in $\mathcal{H}$ such that $\sigma_{ess}(A) = \sigma_{ess}(\widetilde{A}) = [0, \infty)$ and assume $A, \widetilde{A} \geq c$ for some $c \leq 0$. Denote $\alpha = 1 - c$ and let $\delta$ be such that $0 < \delta < \frac{1}{2\alpha}$. Further, assume that
\begin{align*}
\| (A+\alpha)^{-1} - (\widetilde{A} + \alpha)^{-1}\| < \delta
\end{align*}
and
\begin{align*}
\Big ( \frac{E - 2\delta \alpha (E + \alpha)}{1 + 2\delta (E + \alpha)}, \frac{E + 2\delta \alpha (E + \alpha)}{1 - 2\delta (E + \alpha)} \Big ) \cap \sigma(A) = \{ E \},
\end{align*}
where $E$ is an eigenvalue of $A$ with multiplicity $n \in \mathbb{N}$. Then,
\begin{align*}
\Big ( \frac{E - \alpha \delta (E+\alpha)}{1 + \delta (E+\alpha)}, \frac{E + \alpha \delta (E + \alpha)}{1 - \delta(E + \alpha)} \Big ) \cap \sigma(\widetilde{A}) \subset \sigma_d(\widetilde{A})
\end{align*}
consists of finitely many eigenvalues of $\widetilde{A}$ and the number, counting multiplicities, of these eigenvalues is equal to $n$.
\end{thm}
\begin{proof}
Apply the previous lemma to the resolvents $(A + \alpha)^{-1}$ and $(\widetilde{A} - \alpha)^{-1}$ and use the spectral mapping theorem.
\end{proof}
In many applications the minimum of the spectrum of $A$ is an isolated eigenvalue with multiplicity one. Hence it follows that the minimum of the spectrum
of $A_n$ is an isolated eigenvalue with multiplicity one for eventually every $n$, if the sequence $(A_n)$ converges to $A$ in the norm resolvent sense. 
Via the min-max-principle the rate of convergence of the lowest eigenvalue can be estimated directly with the aid of the associated sesquilinear forms:
\begin{thm}
Let $A$ and $B$ be self-adjoint and lower semibounded linear operators in $\mathcal{H}$ such that $A \geq 1-c$ and $B \geq 1-c$ with $c > 0$.
Let $a$ and $b$ be the closed sesquilinear forms associated to $A$ and $B$, respectively. Assume that 
%Further assume that, letting $a$ be the closed sesquilinear form associated to $A$ and $b$ be the closed sesquilinear form associated to $B$, 
$D(a) = D(b) =: D$ and 
\[ s := \sup_{g \in D, \  a_c(g, g) = 1} |a(g, g)  - b(g, g)| < 1.\]
Further assume that $E_1(A) := \min \sigma(A)$ is a negative eigenvalue of $A$ and $E_1(B) := \min \sigma(B)$ is a negative eigenvalue of $B$. Then
\begin{eqnarray}\label{ground_state_energy}
E_1(B) \le E_1(A) + (E_1(A)+c)s.
\end{eqnarray}
\end{thm}
\begin{proof}
Let $f$ be a normalized eigenvector of $A$ corresponding to the eigenvalue $E_1(A)$. Then $a_c(f,f) = (E_1(A)+c)$. 
Put
\[ g: = \frac{f}{\sqrt{E_1(A) +c}}.\]
Then $a_c(g,g) = 1$. By the min-max-principle, 
\begin{eqnarray*}
\min \sigma(B+cI) & = & \min_{h\in D,\parallel h \parallel =1} b_c(h,h) \\
& \le & b_c(f,f) = (E_1(A) +c) b_c(g,g) \\
& \le & (E_1(A) +c) (a_c(g,g)+s) = (E_1(A)+c)(1+s),
\end{eqnarray*}
and hence $E_1(B)=\min \sigma(B) \le E_1(A) + (E_1(A)+c)s \le E_1(A) +cs$. 
\end{proof}
\begin{rem} a) Of course, changing the roles of $A$ and $B$ we also get a lower bound for $E_1(B)$. \\
b) The smaller $c$ is the better is the estimate (\ref{ground_state_energy}).
In general the smallest $c$ is not known. However, if one knows (approximately) $E_1(A)$ and any
constant $c$ satisfying the hypothesis of the previous 
theorem and $cs$ is sufficiently small, then one may use (\ref{ground_state_energy}) repeatedly in order to get smaller and smaller 
constants $c$ in the estimate (\ref{ground_state_energy}) and hence better estimates for $E_1(B)$.
\end{rem}
\section{Convergence of the sesquilinear forms}
The aim of this section is to present a quantitative version of the result presented in Theorem \ref{formconvergence} above.
\begin{lemma}[{\cite[Proof of Theorem 3]{Brasche}}]\label{lemma8}
Let $\mu$ and $\mu_n$, $n \in \mathbb{N}$, be finite Radon measures on $\mathbb{R}$ such that the sequence $(\mu_n)$ converges weakly to $\mu$. Then there exists a common lower bound $c$ for $a_\mu$ and all $a_{\mu_n}$ and
\begin{align*}
\sup_{g \in H^1(\mathbb{R}), (a_\mu)_{1-c}(g,g) = 1} |(a_{\mu} - a_{\mu_n})(g,g)| \leq 2 \sup_{g \in H^1(\mathbb{R}), \| g\|_{H^1(\mathbb{R})} \leq 1} |(a_{\mu} - a_{\mu_n})(g,g)|. 
\end{align*}

\end{lemma}
\begin{prop}[{\cite[Lemma 2]{Brasche}}]\label{formestimate}
Let $\nu$ and $\mu$ be finite Radon measures on $\mathbb{R}$. Then, for each $g \in H^1(\mathbb{R})$,
\begin{align*}
|a_\mu(g,g) - a_{\nu}(g,g)| = \Big | \int_{\mathbb{R}} \vert g \vert^2 d(\mu - \nu) \Big | \leq 
\| g\|_{H^1}^2 \frac{2}{\sqrt{\pi}} \Big (\int_{\mathbb{R}} \frac{1}{1+t^2}|\hat{\mu}(t) - \hat{\nu}(t)|^2 dt \Big )^{1/2}.
\end{align*} 
In particular, if the finite Radon measures $\mu_n$ on $\mathbb{R}$ converge weakly to  $\mu$ and 
the forms $a_{\mu_n}$  have a common lower bound $c$, then the forms $a_{\mu_n}$ converge to the form $a_\mu$, and
\begin{align*}
\sup_{g \in H^1(\mathbb{R}), (a_{\mu})_{1-c}(g,g) = 1} |a_{\mu}(g,g) - a_{\mu_n}(g,g)| \leq \frac{4}{\sqrt{\pi}} \Big ( \int_{\mathbb{R}} \frac{1}{1+t^2} |\hat{\mu}(t) - \hat{\mu}_n(t)|^2 dt \Big )^{1/2} \to 0
\end{align*}
for $n \to \infty$.
\end{prop}
The proposition provides an upper bound for the error one makes by truncating the potential. 
\begin{cor}\label{trunc}
Let $\mu$ be a finite Radon measure on $\R$ and $B\subset \R$ be a Borel set. Then
\begin{eqnarray}\label{truncation}
\vert a_{\mu}(g,g) - a_{\chi_B \mu}(g,g) \vert \le 2 \| g \|_{H^1}^2 \vert \mu \vert (\R \setminus B)
\end{eqnarray}
for every $g\in H^1(\R)$. 
\end{cor}
\begin{proof} For every $t\in \R$ 
\[ \vert \hat{\mu}(t) - \widehat{\chi_B \mu}(t) \vert = \vert \int_{\R} e^{itx} 1_{\R\setminus B}(x) d\mu(x) \vert
\le \vert \mu \vert (\R\setminus B), \]
and hence the corollary follows from the previous theorem with $\nu=1_B\mu$.  
\end{proof}

%\section{Truncating the potential and norm resolvent estimates}\label{trunc}

\section{Weak approximation of measures by pure point measures}\label{approximationofmeasures}
Given a finite Radon measure $\mu$ on $\mathcal{B}(\mathbb{R})$, we want to give a constructive way to approximate it by pure point measures. For this it is justified by Corollary \ref{trunc} to assume that $\mu$ has compact support. One can of course decompose $\mu$ into a continuous and a discrete part,
\[ \mu = \mu_c + \mu_d \]
where $\mu_d = \sum_{j=1}^\infty \alpha_j \delta_{x_j}$ for pairwise different $x_j \in \mathbb{R}$, $(\alpha_j) \in \ell^1(\mathbb{N}, \mathbb{R})$ and $\mu_c$ has a continuous cumulative distribution function $F_{\mu_c}(x) = \mu_c((-\infty, x])$. Both parts will be approximated separately. First observe that, using $\mu_d^n =  \sum_{j=1}^n \alpha_j \delta_{x_j}$, one easily sees that $\mu_d^n$ converges weakly to $\mu_d$ and the following error estimate for the Fourier transforms of the measures follows directly:
\begin{align*}
|\hat{\mu}_d(t) - \hat{\mu}_d^n(t)| \leq \sum_{j=n+1}^\infty |\alpha_j|
\end{align*}
Now let us assume that the measure $\mu$ is purely continuous. Further, 
assume  that $\mu$ is supported inside the interval $[-K, K]$ for some sufficiently large $K>0$. For convenience we may assume that $K \in \mathbb{N}$. For $N \in \mathbb{N}$ set
\begin{align*}
x_j^N &:= -K + \frac{j}{N}, \quad j = 0, \dots, 2NK,\\
a_j^N(\mu) &:= F_\mu(x_j) - F_\mu(x_{j-1}) = \mu(x_{j-1}, x_j), \quad j = 1, \dots, 2NK
\end{align*}
and
\[ \mu_N := \sum_{j=1}^{2NK} a_j^N(\mu) \delta_{x_j^N}. \]
It is easy to verify that $F_{\mu_N}$ converges pointwise to $F_{\mu}$ and hence $\mu_N$ converges weakly to $\mu$. The following proposition gives an error estimate in terms of the Fourier transforms:
\begin{prop}
Let $\mu$ be a finite Radon measure on $\mathbb{R}$ with $\operatorname{supp}(\mu) \subset [-K, K]$ and continuous cumulative distribution function and let $\mu_N$ be constructed as above. Let $t \in \mathbb{R} \setminus \{ 0\}$ and $\varepsilon > 0$. Choose  $N \in \mathbb{N}$ such that $\frac{1}{N}< \frac{\pi}{2|t|}\min \{ 1, \frac{\varepsilon^2}{2}\}$. Then
\begin{align*}
|\widehat{\mu}(t) - \widehat{\mu}_N(t)| \leq \varepsilon|\mu|(\mathbb{R}).
\end{align*}
\end{prop}
\begin{proof}
First observe that, for given $t \neq 0$, $\delta_\varepsilon(t) = \frac{\pi}{2|t|} \min\{1, \frac{\varepsilon^2}{2}\}$ is a possible choice of $\delta$ for the uniform continuity of the function $x \mapsto e^{itx}$ and given $\varepsilon > 0$. The result follows from the following computations:
\begin{align*}
|\widehat{\mu}(t) - \widehat{\mu}_N(t)| &= \Big | \int_{\mathbb{R}} e^{itx} d\mu(x) - \int_{\mathbb{R}} e^{itx} d\mu_N(x) \Big |\\
&= \Big | \int_{\mathbb{R}} e^{itx} d\mu(x) - \sum_{j=1}^{2NK} \mu((x_{j-1}, x_j))e^{itx_j} \Big |\\
&= \Big | \sum_{j=1}^{2NK} \int_{(x_{j-1}, x_j)} e^{itx} - e^{itx_j} d\mu(x) \Big |\\
&\leq \sum_{j=1}^{2NK} \int_{(x_{j-1},x_j)} \underbrace{|e^{itx} - e^{itx_j}|}_{\leq \varepsilon} d|\mu|(x) \\
&\leq \varepsilon |\mu|(\mathbb{R}).
\end{align*}
\end{proof}
Given any $\tilde{K} > 0$ and any $\delta >0$, one can choose $N$ such that 
$|\hat{\mu}(t) - \hat{\mu}_N(t)|< \delta$ for every $t\in [-\widetilde{K},\widetilde{K}]$. In fact, in the preceding proposition  
one simply has to choose $\varepsilon >0$ sufficiently small and then $N$ sufficiently large. 
For the final error estimate, we need not to give bounds for $|\hat{\mu}(t) - \hat{\mu}_N(t)|$, but for the integral occuring in Proposition \ref{formestimate}. This bound is discussed now. Observe that the assumption $|\mu_N|(\mathbb{R}) \leq |\mu|(\mathbb{R})$ is always fulfilled for our construction.
\begin{prop}
Let $\mu$ and $\nu$ be finite Radon measures on $\mathbb{R}$ with $|\nu|(\mathbb{R}) \leq |\mu|(\mathbb{R})$. Let 
$\widetilde{K}$ and $ \varepsilon $ be any positive real numbers and suppose, in addition, that 
$|\hat{\mu}(t) - \hat{\nu}(t)| < \varepsilon$ for every $t$ in the interval $[-\widetilde{K}, \widetilde{K}]$. Then 
\begin{align*}
\int_\mathbb{R} \frac{1}{1+t^2}|\hat{\mu}(t) - \hat{\nu}(t)|^2 dt \leq \varepsilon^2 \pi + 8|\mu|(\mathbb{R})^2 \arctan \Big (\frac{1}{\widetilde{K}} \Big ).
\end{align*} 
\end{prop}
\begin{proof}
We use
\begin{align*}
\int_\mathbb{R}\frac{1}{1+t^2}|\hat{\mu}(t) - \hat{\nu}(t)|^2 dt &= \quad \int_{-\infty}^{-\widetilde{K}}\frac{1}{1+t^2}|\hat{\mu}(t) - \hat{\nu}(t)|^2 dt \\
& \quad + \int_{-\widetilde{K}}^{\widetilde{K}}\frac{1}{1+t^2}|\hat{\mu}(t) - \hat{\nu}(t)|^2 dt \\
&\quad + \int_{\widetilde{K}}^{\infty}\frac{1}{1+t^2}|\hat{\mu}(t) - \hat{\nu}(t)|^2 dt
\end{align*}
and estimate the three integrals separately, where the estimates for the first and the third integral are identical. For the first integral
\begin{align*}
\int_{-\infty}^{-\widetilde{K}}\frac{1}{1+t^2}|\hat{\mu}(t) - \hat{\nu}(t)|^2 dt &\leq (2|\mu|(\mathbb{R}))^2 \int_{-\infty}^{-\widetilde{K}} \frac{1}{1+t^2}dt\\
&= 4|\mu|(\mathbb{R})^2 (\arctan(-\widetilde{K}) + \frac{\pi}{2})\\
&= 4 |\mu|(\mathbb{R})^2 \arctan \Big ( \frac{1}{\widetilde{K}} \Big )
\end{align*}
and for the second integral
\begin{align*}
\int_{-\widetilde{K}}^{\widetilde{K}}\frac{1}{1+t^2}|\hat{\mu}(t) - \hat{\nu}(t)|^2 dt &\leq \varepsilon^2 \int_{-\widetilde{K}}^{\widetilde{K}} \frac{1}{1+t^2} dt\\
&\leq \varepsilon^2 \pi.
\end{align*}
\end{proof}

\section{Eigenvalues of \texorpdfstring{$-\Delta + \sum_j \alpha_j \delta_{x_j}$}{-Delta + sum alpha j delta xj}}\label{sec8}
Let $x_1 < x_2 < \hdots < x_k$ and $\alpha_j \in \mathbb{R}$ for $j=1,2,\ldots,k$. 
As it is well known, the domain of the operator $-\Delta + \sum_{j=1}^k \alpha_j \delta_{x_j}$ (using the notation $x_0 = -\infty$ and $x_{k+1} = \infty$)
is given by
\begin{align*}
D(-\Delta + \sum_{j=1}^k \alpha_j \delta_{x_j}) = \Big \{ f& \in H^1(\mathbb{R}) \cap  \bigoplus_{j=0}^{k} H^2((x_j, x_{j+1})); \\
&f'(x_j+) - f'(x_j-) = \alpha_j f(x_j) \text{ for } j=1, \dots, k \Big \}
\end{align*}
and the operator is just acting as $f \mapsto -f''$. As it can easily be seen, such operators can only have negative eigenvalues. For the readers convenience we will now repeat an algorithm to find the eigenvalues of such an operator (cf. \cite[proof of Theorem II 2.1.3]{Albeverio}). Each eigenfunction $f_\lambda$ for an eigenvalue $\lambda < 0$ of the operators needs to be of the form
\begin{align*}
f_\lambda(x) = a_j e^{\sqrt{-\lambda}x} + b_j e^{-\sqrt{-\lambda}x}, \quad x_j < x < x_{j+1}
\end{align*}
for some constants $a_j, b_j \in \mathbb{R}$. Since the function needs to be in $L^2(\mathbb{R})$, it is necessary that $b_0 = 0$. By linearity we may also assume $a_0 = 1$ (since $a_0 = 0$ would imply $f_\lambda = 0$). Such a function is in the domain of the operator (and hence an eigenfunction) if and only if it satisfies all the continuity and $\delta$-boundary conditions, i.e.
\begin{align*}
f_\lambda (x_j+) &= f_\lambda (x_j-),\\
f_\lambda'(x_j+) - f_\lambda'(x_j-) &= \alpha_j f_\lambda(x_j),
\end{align*}
for $j = 1, \dots, k$ (and $a_k = 0$, which we will ignore for a short moment). Starting with $a_0 = 1, b_0 = 0$, these two conditions give a $2 \times 2$ system of linear equations at $x_1$ for $a_1, b_1$, which one can easily solve. Continuing with the conditions at $x_2$, one gets equations for $a_2, b_2$, and so on. In the end, one computes all the $a_j$ and $b_j$ such that the continuity and boundary-conditions are automatically fulfilled. However, in general one will obtain $a_k \neq 0$. One can easily see that $\lambda$ is an eigenvalue if and only if the $a_k$ obtained by this method is equal to 0 
(only then we have  $f_\lambda \in L^2(\mathbb{R})$). Hence, we consider the $a_k$ computed in the above manner as a function of $\lambda$. This $a_k(\lambda)$ depends continuously on $\lambda < 0$. The problem of finding a negative eigenvalue of $- \Delta + \sum\alpha_j \delta_{x_j}$ reduces to finding the zero of the continuous real-valued function $a_k(\lambda)$, which may be solved numerically. Observe that the runtime of the evaluation of $a_k(\lambda)$ depends only linearly on $k$ (at each $x_j$ we need to solve a $2\times 2$ system of linear equations).

\section{Existence of negative eigenvalues}
It is crucial for our approximation method that there exist negative eigenvalues of the approximating operator.
Hence we will add a short discussion about the existence of negative eigenvalues for operators $-\Delta + \mu$. Recall (Theorem \ref{essspectrum} above) that $\sigma_{ess}(-\Delta + \mu) = [0, \infty)$ for each finite Radon measure $\mu$. In light of this result, for the existence of a negative eigenvalue of $-\Delta + \mu$ it suffices to show that there is a function $f \in H^1(\mathbb{R})$ with $a_\mu(f, f) < 0$. This directly gives the following well known result, which we quickly prove for completeness.
\begin{prop}\label{theorem8}
If the finite Radon measure $\mu$ satisfies $\mu(\mathbb{R}) < 0$, then  $-\Delta + \mu$ has at least one negative eigenvalue.
\end{prop}
\begin{proof}
For $N \in \mathbb{N}$ consider the function $f_N$ as pictured in Figure 1. For $N$ large enough,
\begin{align*}
a_\mu(f_N, f_N) = \int_\mathbb{R}|f'_N(x)|^2 dx + \int_\mathbb{R}|f_N|^2 d\mu < 0.
\end{align*}
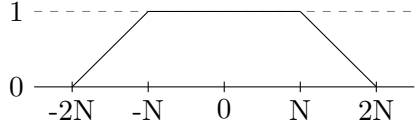
\begin{figure}
\begin{center}
\begin{tikzpicture}
\draw[-] (-2.5,0) to (2.5,0);
\draw[dashed, gray] (-2.5,1) to (2.5,1);
\draw[-] (-2,0) to (-1,1);
\draw[-] (-1,1) to (1,1);
\draw[-] (1,1) to (2,0);
\draw[-] (2,2pt) to (2,-2pt);
\draw[-] (1,2pt) to (1,-2pt);
\draw[-] (-1,2pt) to (-1,-2pt);
\draw[-] (-2,2pt) to (-2,-2pt);
\draw[-] (0,2pt) to (0,-2pt);
\node[anchor=east] at (-2.5,1) {1};
\node[anchor=east] at (-2.5,0) {0};
\node[anchor=north] at (-2,-1pt) {-2N};
\node[anchor=north] at (-1,-1pt) {-N};
\node[anchor=north] at (0,-1pt) {0};
\node[anchor=north] at (1,-1pt) {N};
\node[anchor=north] at (2,-1pt) {2N};
\end{tikzpicture}
\end{center}
\caption{The function $f_N$}
\end{figure}
\end{proof}
\begin{rem}
If $\mu$ is a finite Radon measure with $\mu(\R) < 0$ and $(\mu_n)$ is a sequence of such measures converging weakly to $\mu$, then it is easy to see that $\mu_n(\mathbb{R}) < 0$ for large $n$. In particular, if the approximation is good enough, then there exists a negative approximating eigenvalue of the Schr\"odinger operator $-\Delta + \mu_n$ approximating $-\Delta + \mu$. Therefore, under the assumption $\mu(\R)<0$ our approximation scheme described above always works.
\end{rem}
We will discuss some other results concerning the existence of negative eigenvalues here. First, we will recall some well known results extending the above Proposition \ref{theorem8} in some sense.
\begin{thm}[{\cite[Lemma 8]{Schmincke}}]
If $V \in C(\mathbb{R})$ with $V(x) \to 0$ as $|x| \to \infty$, $V \not \equiv 0$ and $\int_\mathbb{R} V(x) dx \leq 0$, then $-\Delta + V$ has at least one negative eigenvalue.
\end{thm}
\begin{thm}[{\cite[Theorem 2.5]{Simon}}]
Let $V(x)$ be a measurable function such that $\int_\mathbb{R} (1 + |x|^2)|V(x)|dx < \infty$, $V \not \equiv 0$ and $\int_\mathbb{R} V(x) dx \leq 0$. Then $-\Delta + V$ has at least one negative eigenvalue.
\end{thm}
It is unknown to the authors whether the above two results extend to one-dimensional Schr\"odinger operators with measure potentials. If similar results for Schr\"odinger operators with point interaction potentials hold, this would increase the applicability of our method. \\
It is easy to construct a continuous function $V \in L^1(\mathbb{R})$ such that $\int_\mathbb{R} V(x) dx = 0$ but $\int_{[-c, c]} V(X)dx > 0$ for all $c > 0$ and the assumptions for one of the above two theorems are fulfilled. Then $-\Delta + V$ has a negative eigenvalue, but Proposition \ref{theorem8} is not sufficient anymore to guarantee the existence of negative eigenvalues of the approximating Schr\"odinger operators. Hence, we will now discuss some other existence criteria for negative eigenvalues of operators $-\Delta + \sum_{j} \alpha_j \delta_{x_j}$ which are not covered by the above results. Although we are particularly interested in the case of finitely many point interactions, the same results hold in case for infinitely many point interactions by the same proofs, i.e if we assume $\alpha_j, x_j \in \mathbb{R}$ for all $j \in \mathbb{Z}$, $\sum_{j = -\infty}^\infty |\alpha_j| < \infty$ and the sequence $(x_j)_j$ is increasingly ordered. Nevertheless, we only formulate the results for finitely many point interactions.
\begin{prop}
Let $\alpha_j, x_j \in \mathbb{R}$ for $j = 1, \dots, n$ and assume $x_1 < x_2 < \hdots < x_n$. Further, assume $a_k < 0$ for some $k \in \{ 1, \dots, n\}$ and denote $d_k^{-} = x_k - x_{k-1}$ and $d_k^{+} = x_{k+1} - x_k$ (with $d_k^{-} = \infty$ for $k=1$ and $d_k^{+} = \infty$ for $k = n$). If
\begin{align*}
\frac{1}{d_k^{-}} + \frac{1}{d_k^{+}} < -\alpha_k,
\end{align*}
then $-\Delta + \sum_{j=1}^n \alpha_j \delta_{x_j}$ has at least one negative eigenvalue.
\end{prop}
\begin{figure}
\begin{center}
\begin{tikzpicture}
\draw[-] (-2.5,0) to (2.5,0);
\draw[dashed, gray] (-2.5,1) to (2.5,1);
\draw (-2,0) to (0,1);
\draw (0,1) to (1.5,0);
\draw (-2,2pt) to (-2,-2pt);
\draw (0,-2pt) to (0, 2pt);
\draw (1.5, -2pt) to (1.5,2pt);
\node[anchor=east] at (-2.5,1) {1};
\node[anchor=east] at (-2.5,0) {0};
\node[anchor=north] at (-2,-1pt) {$x_{k-1}$};
\node[anchor=north] at (0,-1pt) {$x_k$};
\node[anchor=north] at (1.5,-1pt) {$x_{k+1}$};
\end{tikzpicture}
\end{center}
\caption{The function $f_{x_k}$}
\end{figure}
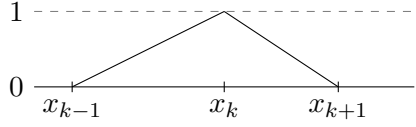
\begin{proof}
For $d_k^{-}, d_k^{+} < \infty$ consider $f_{x_k}$ as pictured in Figure 2. Then
\begin{align*}
a_{\sum_j \alpha_j \delta_{x_j}}(f_{x_k}, f_{x_k}) &= \int_\mathbb{R} |f_{x_k}'(x)|^2 dx + \sum_{j=1}^n \alpha_j |f_{x_k}(x_j)|^2\\
&= \frac{1}{d_k^{-}} + \frac{1}{d_k^{+}} + \alpha_k < 0.
\end{align*}
The cases $d_k^{-} = \infty$ or $d_k^{+} = \infty$ work similarly.
\end{proof}
\begin{prop}
Let $\alpha_j, x_j \in \mathbb{R}$ for $j =1, \dots, n$ and assume $x_1 < \hdots < x_n$. Further, assume $a_k < 0$ for some $k \in \{ 1, \dots, n\}$. If
\begin{align*}
\frac{\alpha_k}{2}e^{-\alpha_k x_k} + \sum_{x_j < x_k} \alpha_je^{-\alpha_k x_j} + e^{-2\alpha_k x_k} \sum_{x_j > x_k} \alpha_j e^{\alpha_k x_j} < 0,
\end{align*}
then $-\Delta + \sum_j \alpha_j \delta_{x_j}$ has at least one negative eigenvalue.
\end{prop}
\begin{proof}
The function
\begin{align*}
f(x) = \begin{cases}
e^{-\frac{\alpha_k}{2} x}, \quad x < x_k\\
e^{-\alpha_k x_k}e^{\frac{\alpha_k}{2}x}, \quad x > x_k
\end{cases}
\end{align*}
is an eigenfunction of $-\Delta + \alpha_k \delta_{x_k}$ to the eigenvalue $-\frac{\alpha_k^2}{4}$. Hence,
\begin{align*}
a_{\sum_j \alpha_j \delta_{x_j}}(f,f) &= \int_{\mathbb{R}} |f'(x)|^2 dx + \alpha_k |f(x_k)|^2 + \sum_{x_j < x_k} \alpha_j |f(x_j)|^2 + \sum_{x_j > x_k} \alpha_j |f(x_j)|^2\\
&= \langle (-\Delta + \alpha_k \delta_{x_k})f, f\rangle + \sum_{x_j < x_k} \alpha_j |f(x_j)|^2 + \sum_{x_j > x_k} \alpha_j |f(x_j)|^2\\
&= -\frac{\alpha_k^2}{4}\| f\|^2 + \sum_{x_j < x_k} \alpha_j |f(x_j)|^2 + \sum_{x_j > x_k} \alpha_j |f(x_j)|^2\\
&= \frac{\alpha_k}{2}e^{-\alpha_k x_k} + \sum_{x_j < x_k} \alpha_je^{-\alpha_k x_j} + e^{-2\alpha_k x_k} \sum_{x_j > x_k} \alpha_j e^{\alpha_k x_j} < 0.
\end{align*}
\end{proof}
\section{Numerical examples}
We want to continue with the presentation of two numerical examples. The following result will be useful to understand the examples, where $N_0(-\Delta + \mu)$ denotes the number of negative eigenvalues of the operator $-\Delta + \mu$ (counting multiplicities).
\begin{thm}[{\cite[Theorem 3.5]{Brasche2}}]\label{n0est}
Let $\mu = \mu_+ - \mu_-$ be a finite Radon measure on $\mathbb{R}$ with corresponding Hahn-Jordan decomposition. Then
\begin{align*}
N_0(-\Delta + \mu) \leq 1 + \frac{1}{2} \frac{\int_\mathbb{R}\int_\mathbb{R} |x-y| d\mu_{-}(x) d\mu_{-}(y)}{\mu_{-}(\mathbb{R})}.
\end{align*}
\end{thm}
\subsubsection*{Square well potential}
The first example is well known, since it is essentially solvable. Hence, it will play the role of a test scenario.\\
Consider the operator $-\Delta - \chi_{[-1,1]}$, i.e. the Schr\"odinger operator with a square well potential. By Proposition \ref{theorem8} and Theorem \ref{n0est}, 
\[ 1 \leq N_0(-\Delta - \chi_{[-1,1]}) \leq 1 + \frac{2}{3}. \]
Further, the only negative eigenvalue $\lambda$ satisfies $-1 < \lambda < 0$ and the equation
\[ \tan \lambda = \frac{\sqrt{1 - \lambda^2}}{\lambda} \]
(see e.g. \cite[Chapter 7.4]{Gustafson}). By solving this equation numerically, one obtains $\lambda \approx -0.453753165860328$.\\
For $N \in \mathbb{N}$ we get the $N$th approximating operator $-\Delta + \sum_{j=1}^{2N} \alpha_j^N \delta_{x_j^N}$ given by
\begin{align*}
x_j^N &= -1 + \frac{j}{N}, \quad j = 1, \dots, 2N;\\
\alpha_j^N &= -\frac{1}{N}, \quad j=1, \dots, 2N.
\end{align*}
The negative eigenvalue for the $N$th approximating operator, found by the procedure described in section \ref{sec8} implemented in Matlab, can be seen in Table 1.
\begin{table}
\begin{center}
\begin{tabular}{|c|c|}
\hline N & Eigenvalue \\ \hline
1 & -0.545877203227244 \\ \hline
2 & -0.474617739449437 \\ \hline
3 & -0.462861650386081\\ \hline
4 & -0.458844821756164\\ \hline
5 & -0.457002447176188\\ \hline
10 & -0.454562375073084\\ \hline
25 & -0.453882500447814\\ \hline
50 & -0.453785494551346\\ \hline
75 & -0.453767533759611\\ \hline
100 & -0.453761247723386\\ \hline 
1000 & -0.453753246677936\\ \hline 
10000 & -0.453753166668506\\ \hline 
100000 & -0.453753165868416\\ \hline 
1000000 & -0.453753165860430\\ \hline \hline
$\lambda$ & -0.453753165860328\\ \hline
\end{tabular}
\end{center}
\caption{Approximation of eigenvalue of $-\Delta - \chi_{[-1,1]}$.}
\end{table}
\subsubsection*{Cantor measure potential}
The second example is supposed to show that our methods also work with respect to rather exotic potentials. Let $\mu_C$ be the Cantor measure, i.e. the measure which has the Cantor function as its cumulative distribution function. We want to find eigenvalues of the operator $-\Delta -\mu_C$. Using again Proposition \ref{theorem8} and Theorem \ref{n0est} one gets
\begin{align*}
1 \leq N_0(-\Delta - \mu_C) \leq 1 + \frac{1}{2},
\end{align*}
i.e. the operator has exactly one negative eigenvalue. While it is possible to approximate this eigenvalue of $-\Delta -\mu_C$ by the method described in section \ref{approximationofmeasures}, we modify the approach here to obtain better results for this particular case. As it is well known, the Cantor set can be obtained as the limit set of a procedure of removing middle third intervals. Inspired by this, one can construct the Cantor measure easily as the limit of a sequence of pure point measures as follows: Let $\mu_N$ be the measures defined by (cf. Figure 3)
\begin{align*}
\Lambda_0 &= \Big \{ \frac{1}{2} \Big \},\\
\Lambda_N &= \Big \{ \frac{x}{3}; x \in \Lambda_{N-1} \Big \} \cup \Big \{ 1 - \frac{x}{3}; x \in \Lambda_{N-1} \Big \}, \quad N \in \mathbb{N}\\
\mu_N &= \frac{1}{2^N}\sum_{x \in \Lambda_N} \delta_{x}, \quad N \in \mathbb{N}.
\end{align*}
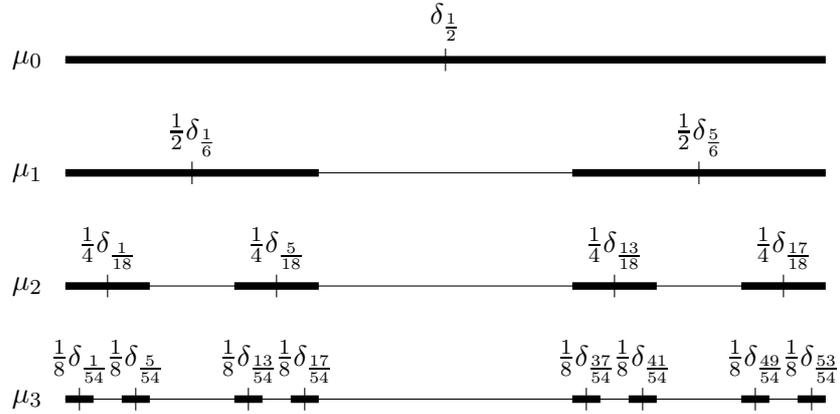
\begin{figure}
\begin{center}
\begin{tikzpicture}
\draw (0,0) -- (10,0);
\draw (0,-1.5) -- (10,-1.5);
\draw (0,-3) -- (10,-3);
\draw (0,-4.5) -- (10,-4.5);
\draw (-0.5,0) node {$\mu_0$};
\draw (-0.5,-1.5) node {$\mu_1$};
\draw (-0.5,-3) node {$\mu_2$};
\draw (-0.5,-4.5) node {$\mu_3$};
\draw (5,0.5) node {$\delta_{\frac{1}{2}}$};
\draw(5,0.15) -- (5,-0.15);
\draw (10/6,-1.5+0.15) -- (10/6,-1.5-0.15);
\draw (50/6,-1.5+0.15) -- (50/6,-1.5-0.15);
\draw (10/6,{-0.5-0.5}) node {$\frac{1}{2} \delta_{\frac{1}{6}}$};
\draw (10-10/6,{-0.5-0.5}) node {$\frac{1}{2} \delta_{\frac{5}{6}}$};
\draw (10/18,{-3+0.15}) -- (10/18,{-3-0.15});
\draw (50/18,{-3+0.15}) -- (50/18,{-3-0.15});
\draw (130/18,{-3+0.15}) -- (130/18,{-3-0.15});
\draw (170/18,{-3+0.15}) -- (170/18,{-3-0.15});
\draw (10/18,-2.5) node {$\frac{1}{4} \delta_{\frac{1}{18}}$};
\draw (20/9 + 10/18,-2.5) node {$\frac{1}{4} \delta_{\frac{5}{18}}$};
\draw (10 - 10/18,-2.5) node {$\frac{1}{4} \delta_{\frac{17}{18}}$};
\draw (10 - 20/9 - 10/18,-2.5) node {$\frac{1}{4} \delta_{\frac{13}{18}}$};
\draw (10/54,-4.5-0.15) -- (10/54,-4.5+0.15);
\draw (50/54,-4.5-0.15) -- (50/54,-4.5+0.15);
\draw (130/54,-4.5-0.15) -- (130/54,-4.5+0.15);
\draw (170/54,-4.5-0.15) -- (170/54,-4.5+0.15);
\draw (370/54,-4.5-0.15) -- (370/54,-4.5+0.15);
\draw (410/54,-4.5-0.15) -- (410/54,-4.5+0.15);
\draw (490/54,-4.5-0.15) -- (490/54,-4.5+0.15);
\draw (530/54,-4.5-0.15) -- (530/54,-4.5+0.15);
\draw (10/54,-4) node {$\frac{1}{8} \delta_{\frac{1}{54}}$};
\draw (10/54 + 20/27,-4) node {$\frac{1}{8} \delta_{\frac{5}{54}}$}; 
\draw (10/54 + 60/27,-4) node {$\frac{1}{8} \delta_{\frac{13}{54}}$};
\draw (10/54 + 80/27,-4) node {$\frac{1}{8} \delta_{\frac{17}{54}}$};
\draw (10 - 10/54,-4) node {$\frac{1}{8} \delta_{\frac{53}{54}}$};
\draw (10 - 10/54 - 20/27,-4) node {$\frac{1}{8} \delta_{\frac{49}{54}}$};
\draw (10 - 10/54 - 60/27,-4) node {$\frac{1}{8} \delta_{\frac{41}{54}}$};
\draw (10 - 10/54 - 80/27,-4) node {$\frac{1}{8} \delta_{\frac{37}{54}}$};
\draw[line width = 0.1cm] (0,0) -- (10,0);
\draw[line width = 0.1cm] (0,-1.5) -- (10/3,-1.5);
\draw[line width = 0.1cm] (20/3,-1.5) -- (10,-1.5);
\draw[line width = 0.1cm] (0,-3) -- (10/9,-3);
\draw[line width = 0.1cm] (20/9,-3) -- (10/3,-3);
\draw[line width = 0.1cm] (20/3,-3) -- (20/3 + 10/9,-3);
\draw[line width = 0.1cm] (20/3 + 20/9,-3) -- (10,-3);
\draw[line width = 0.1cm] (0,-4.5) -- (10/27,-4.5);
\draw[line width = 0.1cm] (20/27,-4.5) -- (30/27,-4.5);
\draw[line width = 0.1cm] (10/3 - 30/27,-4.5) -- (10/3 - 20/27,-4.5);
\draw[line width = 0.1cm] (10/3 - 10/27,-4.5) -- (10/3,-4.5);
\draw[line width = 0.1cm] (10,-4.5) -- (10 - 10/27,-4.5);
\draw[line width = 0.1cm] (10 - 20/27,-4.5) -- (10 - 30/27,-4.5);
\draw[line width = 0.1cm] (10 - 10/3 + 30/27,-4.5) -- (10 - 10/3 + 20/27,-4.5);
\draw[line width = 0.1cm] (10 - 10/3 + 10/27,-4.5) -- (10 - 10/3,-4.5);
\end{tikzpicture}
\end{center}
\caption{Construction of the measures $\mu_N$ for $N = 0, 1, 2, 3$. The thick black bars are the intervals in the $N$th step of the construction of the Cantor set.}
\end{figure}
It is easy to show that the $\mu_N$ converge to $\mu_C$ weakly for $N \to \infty$, e.g. by proving that the cumulative distribution functions converge pointwise. Hence, the negative eigenvalues of the operators $-\Delta -\mu_N$ converge to the eigenvalue of $-\Delta - \mu_C$. 
The numerically computed eigenvalues of $-\Delta - \mu_N$ can be seen, for some $N$, in Table 2. We will not give a full discussion of the error estimate of this modified method here, but we will compute $\hat{\mu}_C(t) - \hat{\mu}_N(t)$ in the appendix, which is the crucial step needed for deriving an error estimate (beside the general theory developed above).
\begin{figure}
\begin{center}
\begin{tabular}{|c|c|}
\hline N & Eigenvalue of $-\Delta - \mu_N$ \\ \hline
1 & -0.25 \\ \hline
2 & -0.190826516988754 \\ \hline
3 & -0.182601523317952\\ \hline
4 & -0.181236785438422\\ \hline
5 & -0.181005430450725\\ \hline
10 & -0.180958390783868\\ \hline
15 & -0.180958384580303\\ \hline
20 & -0.180958384579497\\ \hline
\end{tabular}
\end{center}
\caption{Approximation of eigenvalue of $-\Delta-\mu_C$.}
\end{figure}
\section{Extension to other 1-dimensional domains}
In this section we will discuss how to extend Theorem \ref{formconvergence} and our approximation method to Schr\"odinger operators on $L^2([0, \infty))$.  Essentially the same ideas work for Schr\"odinger operators on $L^2(0, 1)$ or to some extend even on finite metric graphs (cf. \cite{Berkolaiko} for an introduction to this topic), but we will only discuss the case $[0, \infty)$ in detail.\\
As it is well known, the self-adjoint realizations of the Laplacian on $L^2([0, \infty))$ can be parametrized by $\alpha \in [0, \pi)$ through the boundary conditions
\begin{align*}
D(-\Delta_\alpha) &= \{ f \in H^2([0, \infty)); \ \cos(\alpha) f(0) + \sin(\alpha)f'(0) = 0\},\\
-\Delta_\alpha f &= -f''.
\end{align*}
The sesquilinear form $a_\alpha^0$ associated to $-\Delta_\alpha$ is given by
\begin{align*}
D(a_\alpha^0) &= H^1([0, \infty))\\
a_\alpha^0(f,g) &= \int_{[0, \infty)} f'(x) \overline{g'(x)} dx - \cot(\alpha) f(0) \overline{g(0)}
\end{align*}
for $\alpha \neq 0$ and
\begin{align*}
D(a_0^0) &= H^1_0([0, \infty))\\
a_0^0(f,g) &= \int_{[0,\infty)} f'(x) \overline{g'(x)}dx.
\end{align*}
Throughout this section we assume that all occuring measures are finite Radon measures on $[0, \infty)$, i.e. finite signed measures on $\mathcal{B}([0, \infty))$. Without loss of generality we may assume that the measures have no point mass at 0 (this would only change the boundary condition and not the action of the resulting Schr\"{o}dinger operator). For such a measure $\mu$ we define the Schr\"odinger operator with boundary conditions corresponding to $\alpha \in [0, \pi)$ and potential $\mu$, $-\Delta_\alpha + \mu$, as the operator associated to the form $a_\alpha^\mu$
\begin{align*}
D(a_\alpha^\mu) &= D(a_\alpha^0),\\
a_\alpha^\mu (f,g) &= a_\alpha^0(f,g) + \int_{[0, \infty)} f \overline{g} d\mu
\end{align*}
in the sense of Kato's first representation theorem, where it can easily be seen that $a_\alpha^\mu$ is a small form perturbation of $a_{\pi/2}^0$ (or of $a_0^0$ if $\alpha = 0$) using the Sobolev inequality
\begin{align}
\forall f \in H^1([0, \infty)) \ \forall \varepsilon > 0: \| f\|_\infty^2 \leq \varepsilon \| f'\|_{L^2}^2 + \frac{4}{\varepsilon} \| f\|_{L^2}^2.
\end{align} Of course the term \textit{boundary conditions} should not be taken too serious in this setting - e.g. if $\mu$ has a singular continuous part in a neighbourhood of $0$, it is completely unclear if the functions in the domain actually fulfill the boundary condition. In what follows we say, in natural analogy to the case of measures on $\mathbb{R}$, that measures $\mu_n$ on $[0, \infty)$ converge weakly to the measure $\mu$ on $[0, \infty)$ if
\begin{align*}
\int_{[0, \infty)} f d\mu_n \to \int_{[0, \infty)} fd\mu
\end{align*}
for all bounded and continuous functions $f$ on $[0, \infty)$.
\begin{thm}
Let $\mu_n$, $n \in \mathbb{N}$ and $\mu$ be finite Radon measures on $[0, \infty)$ without point mass at $0$. If the $\mu_n$ converge weakly to $\mu$, then $-\Delta_\alpha + \mu_n$ converge to $-\Delta_\alpha + \mu$ in the norm resolvent sense.
\end{thm}
\begin{proof}
Using the general theory, it suffices again to prove that the corresponding sesquilinear forms converge. Further, it can be shown in the same manner as on $\mathbb{R}$ (cf. Lemma \ref{lemma8} above and the proof in \cite{Brasche}) that we only need to show that
\begin{align*}
\sup_{f \in D(a_\alpha), \| f\|_{H^1([0, \infty))} \leq 1} \Big | \int_{[0, \infty)} |f|^2 d(\mu - \mu_n) \Big | \to 0
\end{align*}
for $n \to \infty$.
Now observe that each function $f \in H^1([0, \infty))$ with $\| f\|_{H^1([0, \infty))} \leq 1$ can be continued to a function $\widetilde{f} \in H^1(\mathbb{R})$ such that $\| \widetilde{f}\|_{H^1(\mathbb{R})} \leq c$, where $c$ is a constant independend of $f$. One way of achieving this is by setting
\begin{align*}
\widetilde{f}(x) = \begin{cases}
f(x), \quad x \geq 0\\
f(0)x + f(0), \quad -1 \leq x < 0\\
0, \quad x < -1
\end{cases}
\end{align*}
and $c = \sqrt{\frac{11}{3}}$:
\begin{align*}
\| \widetilde{f}\|_{H^1(\mathbb{R})}^2 &= \| f\|_{H^1([0, \infty))}^2 + \int_{-1}^0 |f(0)x + f(0)|^2 dx + \int_{-1}^0 |f(0)|^2 dx\\
&\leq 1 + \frac{4}{3}|f(0)|^2 \leq 1 + \frac{4}{3}\| f\|_\infty^2\\
&\leq 1 + \frac{8}{3} \| f\|_{H^1([0, \infty))}^2 = \frac{11}{3}. 
\end{align*}
For a measure $\nu$ on $[0, \infty)$ denote by $\nu^\ast$ the measure extended to $\mathbb{R}$ by setting it to $0$ on $(-\infty, 0)$. Then
\begin{align*}
\sup_{f \in D(a_\alpha), \| f\|_{H^1([0, \infty))} \leq 1} \Big | \int_{[0, \infty)} |f|^2 d(\mu - \mu_n) \Big | &= \sup_{f \in D(a_\alpha), \| f\|_{H^1([0, \infty))} \leq 1} \Big | \int_{\mathbb{R}} \big |\widetilde{f} \big |^2 d(\mu^\ast - \mu_n^\ast) \Big |\\
&\leq \frac{1}{c^2}\sup_{g \in H^1(\mathbb{R}), \| g\|_{H^1(\mathbb{R})} \leq 1} \Big | \int_{[0, \infty)} |g|^2 d(\mu^\ast - \mu_n^\ast) \Big |.
\end{align*}
As $\mu_n^\ast \to \mu^\ast$ weakly, this finishes the proof since the last quantitiy is known to converge to 0 by the results on $\mathbb{R}$.
\end{proof}
\begin{rem} a) Using the same idea, i.e. extending functions from $H^1([0, 1])$ uniformly to functions in $H^1(\mathbb{R})$, one can prove the corresponding theorem for Schr\"odinger operators on $L^2([0,1])$. Even further, by the same method one gets the same result for Schr\"odinger operators on arbitrary finite metric graphs.\\
b) Using a construction analogous to the one on $\mathbb{R}$, one can easily construct pure point measures $\mu_N$ on $[0, \infty)$ approximating a given finite Radon measure $\mu$ on $[0, \infty)$ weakly.
\end{rem}
It remains to discuss the method of computating the eigenvalues of $-\Delta_\alpha + \sum_{j=1}^n \beta_j \delta_{x_j}$ analogous to the method discussed in section \ref{sec8}. We may assume that the $x_j$ are ordered increasingly, $0 < x_1 < \hdots < x_n$, and $\beta_j  \in \mathbb{R}$. One can check that the operator $-\Delta_\alpha + \sum_j \beta_j \delta_{x_j}$ acts as $f \mapsto -f''$ a.e. and that the functions in its domain fulfill the boundary conditions $\cos(\alpha)f(0) + \sin(\alpha) f'(0) = 0$ and the usual continuity and $\delta$ boundary conditions at all $x_j$. Then, each eigenfunction $f_\lambda$ for an eigenvalue $\lambda < 0$ has the form
\begin{align*}
f_\lambda(x) = a_j e^{\sqrt{-\lambda}x} + b_j e^{-\sqrt{-\lambda}x}, \quad x_{j-1} < x < x_j
\end{align*}
(with $x_0 = 0$ and $x_{n+1} = \infty$) for coefficients $a_j, b_j \in \mathbb{R}, j=0, \dots, n$. As in the case on the whole line, such a function is an eigenfunction (and hence $\lambda$ an eigenvalue) if and only if $a_n = 0$ and it fulfills all boundary conditions. Therefore, let us consider the case $a_0 \neq 0$, from which we may assume $a_0 = 1$. Then the boundary conditions at $0$ uniquely determine $b_0$. We may now continue as we know it from section \ref{sec8}, computing iteratively the values of $a_j$ and $b_j$ from the values of $a_{j-1}, b_{j-1}$ and the boundary and continuity conditions at $x_j$. At the end, we get $a_n = a_n(\lambda)$ as a continuous function of $\lambda$ and only need to find zeroes of this continuous real-valued function.\\
We still need to deal with the case $a_0 = 0$. It is easy to check that for $\alpha = 0, \frac{\pi}{2}$ this never leads to a valid eigenfunction (and hence not to an eigenvalue). Further, if $n = 1$ one can also check that this never results in an eigenvalue. It is unknown to the authors if this case may lead to a negative eigenvalue of $-\Delta_\alpha + \sum_j \beta_j \delta_{x_j}$ for any choice of $\beta_j$ and $x_j$. Anyway, this is not a problem: If $\alpha \neq 0, \frac{\pi}{2}$, then we just consider the case $a_0 = 0, b_0 = 1$ separately. From the boundary conditions, one imediately gets $\lambda = -(\cot(\alpha))^2$, hence we only need to check one more value of $\lambda$. Iterating now over all $x_j$, one can quickly check if this $\lambda$ is an eigenvalue (again, by checking if $a_n(\lambda) = 0$).
\begin{rem}
a) With the same idea, it is possible to find eigenvalues of $-\Delta + \sum_j \alpha_j \delta_{x_j}$ on $L^2([0,1])$ (with suitable boundary conditions). Of course on can (and will) in this case also get eigenvalues $\geq 0$. Here, one also has to use a suitable Ansatz for the eigenfunctions of non-negative eigenvalues.\\
b) This approach to find eigenvalues of $-\Delta + \sum_j \alpha_j \delta_{x_j}$ will not directly work on most metric graphs. But on certain classes of metric graphs, there are substitutional methods available for computing eigenvalues of Laplacians (and eigenvalues of $-\Delta + \sum_j \alpha \delta_{x_j}$ are just eigenvalues of a Laplacian on a metric graph with a few more vertices). If such a method exists, our method for approximating eigenvalues of $-\Delta + \mu$ and the error estimates work as well.
\end{rem}
\section*{Appendix: Cantor measure and Fourier transform}
Let $\mu_C$ be the Cantor measure and
\begin{align*}
\mu_N = \frac{1}{2^N} \sum_{x \in \Lambda_N} \delta_x
\end{align*}
be the $N$th approximating measure, as described above. We want to compute $\hat{\mu}_C(t) - \hat{\mu}_N(t)$. Denoting
\begin{align*}
S_N = \Big \{ \frac{1}{2} + \sum_{j=1}^N \sigma_j \frac{1}{3^j}; \sigma \in \{ -1, 1\}^N \Big \}
\end{align*}
we are first going to show that $\Lambda_N = S_N$ for $N= 1, 2, 3, \dots$. This will follow by induction. For $N = 1$ the relation is obvious. Hence, assume that $S_N = \Lambda_N$. It suffices to prove $\Lambda_{N+1} \subset S_{N+1}$, since clearly $|\Lambda_{N+1}| = 2^{N+1} = |S_{N+1}|$. Let $x \in \Lambda_{N+1}$. Then either \begin{align*}
x = \frac{y}{3} \quad \text{ or } x = 1 - \frac{y}{3}
\end{align*}
for some $y \in \Lambda_N$. Assume the first case is true (the other case can be dealt with in the same way). Then, since $\Lambda_N = S_N$, for some $\sigma \in \{ -1, 1 \}^N$
\begin{align*}
x &= \frac{\frac{1}{2} + \sum_{j=1}^N \sigma_j \frac{1}{3^j}}{3} = \frac{1}{6} + \sum_{j=1}^N \sigma_j \frac{1}{3^{j+1}}\\
&= \frac{1}{2} - \frac{1}{3} + \sum_{j=1}^N \sigma_j \frac{1}{3^{j+1}} = \frac{1}{2} + \sum_{j=1}^{N+1} \tilde{\sigma}_j \frac{1}{3^j} \in S_{N+1}
\end{align*}
with
\begin{align*}
\tilde{\sigma} = (-1, \sigma_1, \sigma_2, \dots, \sigma_N) \in \{ -1, 1\}^{N+1}
\end{align*}
and therefore $x \in  S_{N+1}$.\\
In what follows, we will need the following trigonometric identity:
\begin{align} \label{cosidentity}
\prod_{j=1}^N \cos(\varphi_j) = \frac{1}{2^N}\sum_{\sigma \in \{ -1, 1\}^N} \cos(\sigma_1 \varphi_1 + \dots + \sigma_N \varphi_N), \quad \text{ for } \varphi_j \in \mathbb{R}, \ j = 1, \dots, N.
\end{align}
The case $N = 2$ of this identity is a direct consequence of the angle sum identity for cosine, the general case follows easily by induction.\\
Now we will compute the Fourier transform of $\mu_N$:
\begin{align*}
\hat{\mu}_N(t) &= \frac{1}{2^N}\sum_{x \in S_N} \hat{\delta}_x(t) = \frac{1}{2^N} \sum_{x \in S_N} e^{itx}\\
&= \frac{1}{2^N}e^{\frac{1}{2}it} \sum_{\sigma \in \{ -1, 1\}^N} e^{it(\sigma_1 \frac{1}{3} + \sigma_2 \frac{1}{3^2} + \dots + \sigma_N \frac{1}{3^N})}\\
&= \frac{1}{2^N}e^{\frac{1}{2} it} \sum_{\sigma \in \{-1,1\}^N, \sigma_N = 1} e^{it (\sigma_1 \frac{1}{3} + \sigma_2 \frac{1}{3^2} + \dots + \sigma_N \frac{1}{3^N})} + e^{-it(\sigma_1 \frac{1}{3} + \sigma_2 \frac{1}{3^2} + \dots +  \sigma_N \frac{1}{3^N})}\\
&= \frac{1}{2^{N-1}}e^{\frac{1}{2}it} \sum_{\sigma \in \{-1, 1\}^{N}, \sigma_N = 1} \cos \Big (t(\sigma_1 \frac{1}{3} + \sigma_2 \frac{1}{3^2} + \dots + \sigma_{N-1} \frac{1}{3^{N-1}} + \sigma_N \frac{1}{3^N}) \Big )\\
&= \frac{1}{2^{N-1}} e^{\frac{1}{2}it} \sum_{\sigma \in \{-1, 1\}^{N}, \sigma_N = 1} \frac{1}{2} \Big (\cos \big (t(\sigma_1 \frac{1}{3} + \sigma_2 \frac{1}{3^2} + \dots + \sigma_{N-1} \frac{1}{3^{N-1}} + \sigma_N \frac{1}{3^N}) \big ) \\
&\quad \quad \quad \quad \quad \quad \quad \quad \quad \quad \quad + \cos \big (-t(\sigma_1 \frac{1}{3} + \sigma_2 \frac{1}{3^2} + \dots + \sigma_{N-1} \frac{1}{3^{N-1}} + \sigma_N \frac{1}{3^N}) \big ) \Big )\\
&= \frac{1}{2^N} e^{\frac{1}{2}it} \sum_{\sigma \in \{-1,1\}^N} \cos \Big (\sigma_1 \frac{t}{3} + \sigma_2 \frac{t}{3^2} + \dots + \sigma_N \frac{t}{3^N} \Big )\\
&= e^{\frac{1}{2}it} \prod_{j=1}^N \cos\Big (\frac{t}{3^j} \Big).
\end{align*}
Here we used formula (\ref{cosidentity}) in the last step. Summarizing, we get
\begin{align*}
\hat{\mu}_N(t) = e^{\frac{1}{2}it} \prod_{j=1}^N \cos \Big (\frac{t}{3^j} \Big ).
\end{align*}
Since the $\mu_N$ converge weakly to $\mu_C$, the Fourier transforms converge pointwise. We get the well-known result 
\begin{align*}
\hat{\mu}_C(t) = e^{\frac{1}{2}it} \prod_{j=1}^\infty \cos \Big ( \frac{t}{3^j} \Big ).
\end{align*}
Observe that it seems, in light of the form of the Fourier transforms, that the most natural way to construct the Cantor measure actually may be as the limit of the measures $\mu_N$. For the difference $\hat{\mu}_C(t) - \hat{\mu}_N(t)$ we get
\begin{align*}
\hat{\mu}_C(t) - \hat{\mu}_N(t) = e^{\frac{1}{2}it} \prod_{j=1}^N \cos \Big ( \frac{t}{3^j} \Big ) \Big (  \prod_{j={N+1}}^\infty \cos \Big ( \frac{t}{3^j} \Big ) -1 \Big )
\end{align*}
which converges fastly to 0 uniformly on compact intervals.

\begin{tabular}{l l}
Johannes F. Brasche & Robert Fulsche\\
\href{mailto:johannes.brasche@tu-clausthal.de}{\Letter johannes.brasche@tu-clausthal.de} & \href{mailto:fulsche@math.uni-hannover.de}{\Letter fulsche@math.uni-hannover.de}\\
Institut f\"{u}r Mathematik & Institut f\"{u}r Analysis\\
Technische Universit\"{a}t Clausthal & Leibniz Universit\"{a}t Hannover\\
Erzstra\ss e 1 & Welfengarten 1\\
38678 Clausthal-Zellerfeld & 30167 Hannover\\
Germany & Germany
\end{tabular}

\bibliographystyle{amsplain}
\bibliography{References}
\end{document}